\documentclass[journal]{IEEEtran}
\ifCLASSINFOpdf
\else
\fi
\usepackage{amsmath,amssymb,amsthm}
\usepackage{graphicx,subfigure,epstopdf}
\usepackage[comma,numbers,square,sort&compress]{natbib}
\usepackage{epstopdf}
\usepackage{epic}
\usepackage{hyperref}
\usepackage{subfigure}
\usepackage{graphicx}
\usepackage{epstopdf}
\usepackage{epsfig}
\usepackage{arydshln}
\usepackage{enumerate}
\usepackage[utf8]{inputenc}
\usepackage[english]{babel}
\usepackage{fancyhdr}
\usepackage{lastpage}
\usepackage{color}
\usepackage{dsfont}
\usepackage{changes}
\usepackage{lineno}

\newcommand{\EQ}{\begin{eqnarray}}

\newcommand{\EN}{\end{eqnarray}}

\newcommand{\EQQ}{\begin{eqnarray*}}

\newcommand{\ENN}{\end{eqnarray*}}

\newcommand{\col}{\mbox{col }}
\newcommand{\diag}{\mbox{diag}}

\newtheorem{thm}{\bf \em{Theorem}}

\newtheorem{lem}{\bf \em{Lemma}}
\newtheorem{rem}{\bf \em{Remark}}
\newtheorem{defi}{Definition}
\newtheorem{prop}{\it Property}
\newtheorem{prob}{\bf \em{Problem}}
\newtheorem{proposi}{\bf \em{Proposition}}
\newtheorem{ass}{\bf \em{Assumption}}

\allowdisplaybreaks
\begin{document}
\onecolumn
{\Large

This paper was originally published in

S. Wang and J. Huang, ``Adaptive leader-following consensus for multiple
Euler-Lagrange systems with an uncertain leader,''
\emph{IEEE Transactions on Neural Networks and Learning Systems},  vol. 30, No.
7, July 2019, pp. 2162-2388.

We have spotted some typos in equations (19), (20), (22), (23), (24), (25b), (33c), (42), (43) and (45a)
where, in some places,  the constant $\mu_2$ should be $\mu_1$.

These typos are now corrected and the modified version is posted here.}


%

 \twocolumn
\title{Adaptive Leader-Following Consensus for Multiple Euler-Lagrange Systems with an Uncertain Leader System}

\author{Shimin~Wang
        and~Jie~Huang,~\IEEEmembership{Fellow,~IEEE}\thanks{This work has been supported in part by the Research Grants Council of the
Hong Kong Special Administration Region under grant No. 14202914 and in part by Projects of Major International (Regional) Joint Research Program NSFC (Grant no. 61720106011).
}
\thanks{Shimin Wang and Jie Huang are with 
 the Department of Mechanical and Automation
Engineering, The Chinese University of Hong Kong, Shatin, N.T., Hong
Kong. E-mail: smwang@mae.cuhk.edu.hk, jhuang@mae.cuhk.edu.hk.

Corresponding
author: Jie Huang, jhuang@mae.cuhk.edu.hk.}}

\maketitle

\begin{abstract}
In this paper, we study the leader-following consensus problem of multiple Euler-Lagrange systems subject to an uncertain leader system. We first establish an adaptive distributed observer for a neutrally  stable linear leader system whose system matrix is not known exactly. Under standard assumptions, this adaptive distributed observer can estimate and pass the leader's state to each follower through the communication network of the system without knowing the leader's system matrix exactly. Under the additional assumption that the leader's state is persistently exciting,
this adaptive distributed observer can also asymptotically learn the parameters of the leader's system matrix. On the basis of this adaptive distributed observer, we further synthesize an adaptive distributed control law to solve our problem via the certainty equivalence principle. Our result allows
the leader-following consensus problem of multiple  Euler-Lagrange systems to be solved even if none of the followers knows the system matrix of the leader system exactly.
\end{abstract}




\section{Introduction}
\IEEEPARstart{M}{multi-agent}  control systems arise from many engineering applications such as consensus \cite{jad2003, OlfatiFax07, LiuX1,lzzcrm2016},
attitude synchronization of rigid spacecraft systems \cite{ren2004decentralized, CaiHuang2016}, containment control \cite{[6]},
cooperative synchronization of multiple robots \cite{cai2016leader, chung2009cooperative}, and event-triggered distributed control \cite{Tabuada1, liu2017cooperative, Yi2017}.
A comprehensive coverage of the research on the multi-agent control systems can be found in some  monographs \cite{QuZ2009, RenW2008, ZhangHW2014} and the recent survey paper \cite{ZhuB2017}.

An important class of multi-agent control systems is the so-called Euler-Lagrange systems, which describes
 robot
manipulators \cite{r19, r18}, rigid body systems \cite{r18}, and so on. In recent years, the interest in studying the consensus of multiple Euler-Lagrange systems  is growing.
There are two types of consensus problems, i.e.,
 the leaderless consensus problem and the leader-following consensus problem. The leaderless consensus problem is to design a distributed control law to make the states and/or outputs of the closed-loop system asymptotically synchronize to a same trajectory determined by the initial condition of the overall system while the leader-following consensus problem aims to
 design a distributed control law to drive the states and/or outputs of every follower subsystem to a prescribed trajectory generated by a leader system.
 The two consensus
problems of multiple Euler-lagrange systems have been
extensively studied under various scenarios.
Reference
\cite{Min2011} studied  a leaderless consensus problem of multiple EL systems
under the assumption that the network topology was
strongly connected, balanced and fixed. Reference \cite{ren2009distributed}
further considered the leaderless consensus problem for multiple EL systems over a directed, connected graph
 by a distributed control law.
 The leader-following consensus problem for multiple EL systems was studied in \cite{chung2009cooperative} assuming that all followers have access to the leader's state.
 Later in \cite{RN316}, under the assumption that the communication topology is static, directed and connected, the leader-following consensus problem was studied for multiple uncertain EL systems via adaptive control. More recently, reference \cite{RN138} studied the problem of global output-feedback tracking for a class of multiple Euler-lagrange systems subject to static connected communication network without velocity measurement, and reference \cite{kzystmlh2017} considered the control of uncertain Euler-lagrange systems subject to full-state constraints.

 The leader-following consensus problem for multiple uncertain EL systems subject to
jointly connected switching networks was studied in \cite{cai2014leader}, and \cite{cai2016leader} by  distributed observer approach and adaptive distributed observer approach, respectively.
A distributed observer is a dynamic compensator that is able to estimate the state of a leader system and pass it to all followers through the communication network of a multi-agent system. It was first proposed in \cite{su2012cooperative, su2012cooperativesw} to handle the cooperative output regulation problem for linear multi-agent systems and then adopted in
\cite{cai2014leader} for dealing with the leader-following consensus problem for multiple EL systems. The distributed observer works under the assumption that every follower knows the system matrix of the leader system. To relax this assumption, a so-called adaptive distributed observer was further proposed in
  \cite{cai2016leader, cai2017adaptive}, which not only estimates the state but also the system matrix of the leader system. As a result, it only requires that the children of the leader know the leader's system matrix and is thus more practical than the distributed observer. For example, in  \cite{[4]},  the adaptive distributed observer
  was used  to develop some optimal control protocols for the distributed output synchronization
problem of leader-follower multi-agent systems.

Either the distributed observer or the adaptive distributed observer assumes the dynamics of the leader system is known by at least some followers. In many practical applications,
the dynamics of the leader system may not be known exactly by any follower. For example, a sinusoidal signal can be generated by a two dimensional linear leader system. The parameters of the leader system
are determined by the frequency of the sinusoidal signal. If the frequency of the sinusoidal signal is unknown, then the system matrix of the leader system is also unknown.

To cater for this more practical scenario,  in this paper, we will
first establish an adaptive distributed observer for a neutrally stable linear  leader system whose system matrix is not known exactly. Under standard assumptions, this adaptive distributed observer can estimate and pass the leader's state to each follower through the communication network of the system without knowing the leader's system matrix exactly. Under the additional assumption that the leader's state is persistently exciting,
this adaptive distributed observer can also asymptotically learn the parameters of the leader's system matrix. On the basis of this adaptive distributed observer, we further synthesize an adaptive distributed control law to solve the leader-following consensus problem of multiple  Euler-Lagrange systems subject to an uncertain leader system via the certainty equivalence principle.
Since the leader's signal in this paper is a multi-tone sinusoidal signal, our result in  this paper allows a multi-tone sinusoidal signal to have arbitrary unknown amplitudes, initial phases and frequencies, thus significantly enhancing the results in
\cite{cai2014leader} and \cite{cai2016leader}.

The rest of this paper is organized as follows: In Section \ref{section1}, we formulate our problem and introduce some standard assumptions. In Section \ref{section2}, we establish an
adaptive distributed observer for a neutrally stable linear leader system whose system matrix is uncertain. In Section  \ref{section4}, we apply the adaptive distributed observer for the unknown leader system to synthesize an adaptive distributed control law for solving
the leader-following consensus problem of multiple  Euler-Lagrange systems subject to an uncertain  leader system.
An example will be given in Section \ref{section5}. Finally, we close this paper in Section \ref{section6}.

\textbf{Notation:} $\otimes$ denotes the Kronecker product of matrices. For $X_1,\cdots,X_k\in \mathds{R}^{n\times m}$, let $\mbox{col }(X_1,\cdots,X_k)=[X_{1}^{T}, \cdots,X_{k}^{T}]^T$, $$\mbox{block diag} (X_1,\cdots,X_k)=\left[
                                                      \begin{array}{ccc}
                                                        X_1 &   &   \\
                                                          & \ddots &   \\
                                                          &   & X_k\\
                                                      \end{array}
                                                    \right].
$$ For any $ x \in \mathds{R}^m$, unless described otherwise, $x_i$ denotes the $i^{th}$ component of $x$.  $$\diag(x)=\left[
                                                      \begin{array}{ccc}
                                                        x_1 &   &   \\
                                                          & \ddots &   \\
                                                          &   & x_m\\
                                                      \end{array}
                                                    \right].$$
For any $x\in \mathds{R}^{m}$ with $m$ an even positive integer, let $\phi:\mathds{R}^m\mapsto \mathds{R}^{ \frac{m}{2} \times m}$ be such that
\begin{eqnarray}\label{skewopen}
  \phi(x)&=&\left[
                           \begin{array}{ccccc}
                              -x_{2} &   x_{1}  & \cdots &  0 & 0  \\
                              \vdots &  \vdots & \ddots &  \vdots &  \vdots \\
                             0  &  0 & \cdots  &  -x_{m} & x_{m-1} \\
                           \end{array}
                         \right].
\end{eqnarray}

\section{Problem Formulation and Assumptions}\label{section1}

As in \cite{cai2016leader}, we consider
a group of $N$ systems described by the following Euler-Lagrange equations:
\begin{equation}\label{MARINEVESSEL1}
  M_i\left(q_i\right)\ddot{q}_i+C_i\left(q_i,\dot{q}_i\right)\dot{q}_i+G_i\left(q_i\right)=\tau_i,~ i = 1, \cdots, N
\end{equation}
where, for $i=1,\cdots,N$, $q_i\in \mathds{R}^n$ is the vector of generalized position vector, $M_i\left(q_i\right)\in \mathds{R}^{n\times n}$ is the symmetric positive definite inertia matrix, $C_i\left(q_i,\dot{q}_i\right)\in \mathds{R}^{n\times n}$ is the coriolis and centripetal forces, and $\tau_i\in \mathds{R}^{n}$ is the control torque.
For $i=1,\cdots,N$,  (\ref{MARINEVESSEL1}) has the following three properties \cite{r18}:
\begin{prop}The inertia matrix $M_i\left(q_i\right)$ is symmetric and uniformly positive definite, \end{prop}
  \begin{prop}For all $a,b\in \mathds{R}^n$, $$M_i\left(q_i\right)a+C_i\left(q_i,\dot{q}_i\right)b+G_i\left(q_i\right)=Y_i\left(q_i,\dot{q}_i,a,b\right)\Theta_i, $$
  where $Y_i\left(q_i,\dot{q}_i,a,b\right)\in \mathds{R}^{n\times q}$ is a known regression matrix and $\Theta_i\in \mathds{R}^{p}$ is a constant vector consisting of the uncertain parameters of (\ref{MARINEVESSEL1}).\end{prop}
 \begin{prop}\label{prop1}$\left(\dot{M}_i\left(q_i\right)-2C_i\left(q_i,\dot{q}_i\right)\right)$ is skew symmetric, $\forall q_i,\dot{q}_i\in \mathds{R}^n$.\end{prop}

The leader's signal $q_0\in \mathds{R}^{n}$, which represents the desired generalized position vector,   is assumed to be generated by the following  exosystem:
\begin{subequations}\label{leaderall}\begin{align}
  \dot{v}&=S(\omega)v\label{leader}\\
  q_0&=Cv
\end{align}
\end{subequations}
where $v\in \mathds{R}^{m}$,  $S(\omega)\in \mathds{R}^{m\times m}$ is the system matrix of $(\ref{leader})$, and is allowed to rely on
some unknown parameter vector $\omega \in \mathds{R}^{l}$, and
$C\in \mathds{R}^{n\times m}$ is a known constant matrix.

As in \cite{cai2016leader}, we view the system composed of (\ref{MARINEVESSEL1}) and (\ref{leaderall}) as a multi-agent system of $N+1$ agents with (\ref{leaderall}) as the leader and the $N$ subsystems of (\ref{MARINEVESSEL1}) as $N$ followers. The network topology of the multi-agent system composed of (\ref{MARINEVESSEL1}) and (\ref{leaderall})  is described by a digraph $\bar{\mathcal{G}}=\left(\bar{\mathcal{V}},\bar{\mathcal{E}}\right)$ with
$\bar{\mathcal{V}}=\{0,\cdots,N\}$ and $\bar{\mathcal{E}}\subseteq \bar{\mathcal{V}}\times \bar{\mathcal{V}}$.
Here the node $0$ is associated
with the leader system (\ref{leaderall}) and the node $i$, $i = 1,\dots,N$,
is associated with the $i$th subsystem of system
(\ref{MARINEVESSEL1}), and, for $i=0,1,\dots,N$, $j=1,\dots,N$, $(i,j) \in
\bar{\mathcal{E}} $ if and only if  $\tau_j$ can use the
information of agent $i$ for control. Let
$\bar{\mathcal{N}}_i=\{j,(j,i)\in \bar{\mathcal{E}}\}$
denote the neighbor set of agent $i$.
Let  ${\cal G}=({\cal V},{\cal E})$  denote
the subgraph of $\bar{\mathcal {G}}$ where ${\cal
V}=\{1,\dots,N\}$, and ${\cal E}\subseteq {\cal V} \times
{\cal V}$ is obtained from $\bar{{\cal E}}$ by removing
all edges between the node 0 and the nodes in ${\cal V}$.

In terms of $\bar{\mathcal{G}}$, we can describe our problem as follows.

\begin{prob}[Leader-following Consensus Problem]\label{ldlesp} Given a system consisting of (\ref{MARINEVESSEL1}) and (\ref{leaderall}), and a  digraph $\mathcal{\bar{G}}$, find a  control law of the following form:
\begin{eqnarray}\label{discon}
  \tau_i &=& f_i\left(q_i,\dot{q}_i,\varphi_i,  \varphi_j-\varphi_i,j\in \mathcal{\bar{N}}_i\right) \nonumber\\
  \dot{\varphi}_i &=& g_i\left(\varphi_i,\varphi_j-\varphi_i,j\in \mathcal{\bar{N}}_i\right),~ i=1,\cdots,N.
\end{eqnarray}
where $\varphi_0 = v$, such that,  for $i=1,\cdots,N$, for any initial condition $q_i(0)$, $\dot{q}_i (0)$, $v(0)$ and $\varphi_i(0)$,  $q_i(t)$, $\dot{q}_i(t)$,  $\varphi_i(t)$ exist and are bounded for all $t\geq0$ and satisfy
$$\lim\limits_{t\rightarrow\infty}\left(q_i\left(t\right)-q_0\left(t\right)\right)=0,~~\lim\limits_{t\rightarrow\infty}\left(\dot{q}_i\left(t\right)-\dot{q}_0\left(t\right)\right)=0.$$
\end{prob}
 Since, for each $i=1,\cdots,N$, $j=0,\cdots,N$, the right hand side of (\ref{discon})  depends on $\varphi_{j} (t)$ only if the $jth$ agent is a neighbor of the $ith$ agent at time $t$, the control law (\ref{discon}) satisfies the communications constraints. Such a control law is called a distributed control law. The specific form of the functions $f_i$ and $g_i$, $i=1,\cdots,N$, will be designed later.

\begin{rem} \label{rem1} Problem \ref{ldlesp} was studied in \cite{cai2014leader, cai2016leader} using the distributed observer approach and the adaptive distributed observer approach, respectively.
 In either \cite{cai2014leader}  or \cite{cai2016leader},  the matrix $S$ is assumed to be known exactly.
However, this assumption may not be desirable. For example,  if the matrix $S$ satisfies Assumption \ref{ass1} to be introduced below,  then the leader's signal is a multi-tone sinusoidal signal.
In this case, assuming the matrix $S$ is known exactly amounts to assuming all the frequencies of the leader's signal are known exactly.
As will be seen later, our result in  this paper will allow us to deal with  multi-tone sinusoidal signals with arbitrary unknown amplitudes, initial phases and frequencies.
\end{rem}

We end this section by listing two assumptions for solving Problem \ref{ldlesp}.

\begin{ass}\label{ass0}$\mathcal{\bar{G}}$ contains a spanning tree with the node 0 as the
root, and $\mathcal{G}$ is undirected.
\end{ass}

\begin{ass}\label{ass1} For all $\omega$, all the eigenvalues of $S(\omega)$ are semi-simple with zero real part.
\end{ass}

\begin{rem} Assumption \ref{ass0} is a quite standard assumption for the leader-following consensus problem for  multi-agent systems subject to a static network.
Assumption \ref{ass1} is also assumed in
\cite{cai2014leader} and \cite{cai2016leader}. Under Assumption \ref{ass1}, without loss of generality, we can always assume that
 $S(\omega)$ is skew-symmetric and takes the following form:
 \begin{eqnarray}\label{skew}
  S(\omega)&=& \left[
              \begin{array}{cc}
               0_{ m_0\times m_0} &   \\
                  & \diag(\omega)\otimes a \\
              \end{array}
            \right],
\end{eqnarray}
where $a=\left[
          \begin{array}{cc}
            0& 1 \\
            -1 & 0 \\
          \end{array}
        \right]
$, $\omega=\col\left(\omega_{01},\cdots,\omega_{0l} \right)\in \mathds{R}^{l}$, $m_0+2 l=m$ and $\omega_{0k} >0$, for $k=1,\cdots, l$.
For convenience, in what follows, we assume $m_0=0$. Thus,  $l=\frac{m}{2}$.
\end{rem}

\section{Adaptive Distributed Observer with an Uncertain Leader}\label{section2}

Let us start with this section by recalling the distributed observer for the leader system (\ref{leaderall}) with $\omega$ a known vector proposed in \cite{su2012cooperative} which takes the following form:
\begin{eqnarray}\label{sdisto}
\dot{\eta}_i &=S(\omega)\eta_i + \mu_1\sum\limits_{j \in \mathcal{\bar{N}}_i}(\eta_j-\eta_i)
\end{eqnarray}
where $\eta_i\in \mathds{R}^m$, $\eta_0=v$ and $\mu_1$ is a positive number. By Lemma 2 of \cite{su2012cooperative}, under Assumptions \ref{ass0} and \ref{ass1}, for $i=1,\cdots,N$, $\lim\limits_{t\rightarrow\infty}\left(\eta_i-v\right)=0$. That is why (\ref{sdisto}) is called the distributed observer for the leader system.
However, (\ref{sdisto}) is not fully distributed in the sense that, for every $i=1,\cdots,N$, (\ref{sdisto}) has to rely on the system matrix $S(\omega)$ of (\ref{leaderall}). To overcome this difficulty, reference \cite{cai2016leader} further proposed the following adaptive distributed observer for the leader system (\ref{leaderall}) with $\omega$ a known vector:
\begin{subequations}\label{hfado}
\begin{align}
\dot{\eta}_i &=S(\omega_i)\eta_i + \mu_{1}\sum_{j \in \mathcal{\bar{N}}_i} (\eta_j-\eta_i),\\
\dot{\omega}_i&=\mu_{2}\left(\sum_{j \in \mathcal{\bar{N}}_i} (\omega_j-\omega_i)\right),\label{hfadob}
\end{align}
\end{subequations}
 where,
 for $i=1,\cdots,N$, $\eta_i\in \mathds{R}^m$, $\omega_i\in \mathds{R}^{l}$ is the estimation of $\omega$, $\omega_0=\omega$, $\eta_0=v$, $\mu_{1}$ and $\mu_{2}$ are some positive numbers. By Lemma 2 of \cite{cai2016leader}, under Assumptions \ref{ass0} and \ref{ass1}, we have both $\lim\limits_{t\rightarrow\infty}\left(\omega-\omega_i\right)=0$ and$\lim\limits_{t\rightarrow\infty}\left(\eta_i-v\right)=0$, $i=1,\cdots,N$. That is why we call (\ref{hfado})  an adaptive distributed observer for the leader system.
 The adaptive distributed observer  (\ref{hfado}) has the property that, for each $i=1,\cdots,N$, (\ref{hfado}) only relies on the information of itself and its neighbors, and is thus a significant enhancement over the distributed observer (\ref{sdisto}).

Nevertheless, for those subsystems which are the children of the leader, the adaptive distributed observer  (\ref{hfado}) still needs to know the leader's frequency $\omega$.
Thus, it is still incapable of dealing with the case where the system matrix
$S (w)$ of the leader system (\ref{leaderall}) depends on an unknown vector $\omega$. As pointed out in Remark \ref{rem1}, in many applications, the system matrix
$S (w)$ is not known exactly. To handle this case, we further define the following distributed dynamic compensator:
\begin{subequations}\label{compensator2}
\begin{align}
\dot{\eta}_i &=S\left(\omega_i\right)\eta_i + \mu_1\sum_{j \in \mathcal{\bar{N}}_i} (\eta_j-\eta_i) \label{compensator2b}\\
\dot{\omega}_{i} &=\mu_2\phi \left ( \sum_{j \in \mathcal{\bar{N}}_i}(\eta_j-\eta_i) \right )\eta_i,~~ i=1,\cdots,N \label{compensator2c}
\end{align}
\end{subequations}
where, $\eta_0=v$, for $i=1,\cdots,N$, $\eta_i\in \mathds{R}^m$,   $\omega_i\in \mathds{R}^{l}$,  $S\left(\omega_i\right)=\left(\diag(\omega_i)\otimes a\right)$, and $\mu_1$ and $\mu_2$ are some positive numbers.

\begin{rem}
What makes  (\ref{compensator2})  interesting is  that it is independent of the unknown vector $\omega$.
If,  for any initial condition,  the solution of (\ref{compensator2}) is bounded over $[0, \infty)$, and
 $\lim\limits_{t\rightarrow\infty}\left(\eta_i-v\right)=0$, $i=1,\cdots,N$, then we call (\ref{compensator2}) an adaptive distributed observer for the uncertain leader (\ref{leaderall}).
\end{rem}


Let $e_{vi} = \sum_{j \in \mathcal{\bar{N}}_i}(\eta_j-\eta_i)$,
$\tilde{\eta}_i=\left(\eta_i-v\right)$ and $\tilde{\omega}_i=\left(\omega_i-\omega\right)$. Then, equation (\ref{compensator2}) can be rewritten as follows:
\begin{subequations}\label{reeq12}
\begin{align}
  \dot{\tilde{\eta}}_i &= S\left(\omega\right) \tilde{\eta}_i + \mu_1 e_{vi}+S\left(\tilde{\omega}_i\right)\eta_i\label{reeq12a}\\
\dot{\tilde{\omega}}_{i} &=\mu_2 \phi(e_{vi})\eta_i,~~i=1,\cdots,N.\label{reeq12b}
\end{align}
\end{subequations}
To put (\ref{reeq12}) in a compact form, let
\begin{eqnarray*}
     \eta &=&\col\left(\eta_1,\cdots,\eta_N\right),~\tilde{\eta}=\col\left(\tilde{\eta}_1,\cdots,\tilde{\eta}_N\right),  \\
     \tilde{\omega}&=&\col\left(\tilde{\omega}_1,\cdots,\tilde{\omega}_N\right),~~ e_{v}=\col(e_{v1},\cdots,e_{vN}),\\
    S_d (\tilde{\omega}) &=&\mbox{block diag }\left(S\left(\tilde{\omega}_1\right),\cdots,S\left(\tilde{\omega}_N\right)\right).
\end{eqnarray*}
Then, we have the following relation:
\begin{equation} \label{eqev}
e_v =-(H\otimes I_m)\tilde{\eta},
\end{equation}
where $H$ consists of the last $N$ rows and the last $N$ columns of the Laplacian matrix  $\mathcal{\bar{L}}$  of the graph $\mathcal{\bar{G}}$ \cite{Hu1}.
It can be verified that
equations (\ref{reeq12}) can be put in the following form:
\begin{subequations}\label{reeq2}
\begin{align}
  \dot{\tilde{\eta}} &=\left(I_N\otimes S\left(\omega\right)-\mu_1 H\otimes I_m\right)\tilde{\eta} + S_d (\tilde{\omega})\eta,\label{reeq2b}\\
\dot{\tilde{\omega}}&=\mu_2 \phi_d(e_{v})\eta,\label{reeq2abbbb}
\end{align}
\end{subequations}
where $\phi_d\left(e_v\right)=\mbox{block diag }\left(\phi(e_{v1}),\cdots,\phi(e_{vN})\right)$.

Before stating our main lemmas, we need to establish the following proposition.

\begin{proposi}\label{Proposition2} For any $z\in \mathds{R}^{l}$ and $x,y\in \mathds{R}^{m}$ with $m = 2 l$,
$$x^TS\left(z\right)y = z^T\phi(x)y,$$
where $S\left(z\right)=\left(\diag(z)\otimes a\right)$ and the matrix function $\phi(\cdot)$ is defined in (\ref{skewopen}).
\end{proposi}
\begin{proof} Denote $z= \col\left(z_{1},\cdots,z_{l}\right)$, $x= \col\left(x_{1},\cdots,x_{m}\right)$ and $y= \col\left(y_1,\cdots,y_m\right)$. From  (\ref{skewopen}), we have

\begin{subequations}\begin{align}
              \phi(x)y &=\left[
                           \begin{array}{ccccc}
                              -x_{2} &   x_{1}  & \cdots &  0 & 0  \\
                              \vdots &  \vdots & \ddots &  \vdots &  \vdots \\
                             0  &  0 & \cdots  &  -x_{m} & x_{m-1} \\
                           \end{array}
                         \right]\left[
                                  \begin{array}{c}
                                    y_1 \\
                                    \vdots \\
                                    y_m \\
                                  \end{array}
                                \right]\nonumber\\
                       &= \left[
                                  \begin{array}{c}
                                   x_{1}y_2 -x_{2} y_1\\
                                    \vdots \\
                                    x_{m-1} y_m -x_{m}y_{m-1}
                                  \end{array}
                                \right],\nonumber\\
                 z^T\phi(x)y  &=\sum_{k=1}^{l}z_{k}\left (x_{2k-1} y_{2k} -x_{2k}y_{2k-1}\right) .\nonumber
\end{align}
\end{subequations}
Since
$$S\left(z\right)= \diag(z)\otimes a= \left(\diag(z)\otimes I_2\right)\left(I_{m_1}\otimes a\right),$$
we have
\begin{subequations}\begin{align}
\left(I_{l}\otimes a\right)y &= \col\left(y_{2},~-y_{1},~\cdots,~y_{m},~-y_{m-1}\right),\nonumber\\
 \left(diag(z)\otimes I_2\right)x&= \col\left(z_{1}x_{1},z_{1}x_{2},\cdots,z_{l}x_{m-1},z_{l}x_{m}\right).\nonumber
\end{align}
\end{subequations}
Hence, we have
 \begin{subequations}\begin{align}
x^TS\left(z\right)y&=x^T \left(\diag(z)\otimes I_2\right)\left(I_{l}\otimes a\right)y\nonumber\\
&=\left(\left(\diag(z)\otimes I_2\right)x\right)^T\left(\left(I_{l}\otimes a\right)y\right)\nonumber\\
&=\sum\limits_{k=1}^{l}z_{k}\left(\left(x_{2k-1}y_{2k}-x_{2k}y_{2k-1}\right)\right)\nonumber\\
 & = z^T\phi(x)y.\nonumber
 \end{align}
\end{subequations}
\end{proof}
We now ready to establish our main technical lemmas.

\begin{lem}\label{lemma4} Consider systems (\ref{leaderall}) and (\ref{reeq2}). Under Assumptions \ref{ass0} and \ref{ass1}, for any $v (0)$, any   $\eta(0)$ and $\omega(0)$, any $\mu_1>0$ and $\mu_2>0$, $\eta(t)$ and $\omega(t) $ exist and are bounded for all $t\geq 0$ and satisfy
\EQ
&&\lim\limits_{t\rightarrow\infty} \tilde{\eta} (t) =0,  \label{eq101a}\\
&&\lim\limits_{t\rightarrow\infty}\dot{\tilde{\omega}}(t)=0,  \label{eq101b}\\
&&\lim\limits_{t\rightarrow\infty}S_d \left(\tilde{\omega}\right)\eta =0.  \label{eq101c}
 \EN
\end{lem}
\begin{proof}
Consider the following Lyapunov function candidate for (\ref{reeq2}):
\begin{eqnarray}\label{reeq3b}
V  =\frac{1}{2}\left(\tilde{\eta}^T\left(H\otimes I_m\right)\tilde{\eta}+\mu_{2}^{-1}\tilde{\omega}^T\tilde{\omega}\right).
\end{eqnarray}
By Lemma 4 of \cite{Hu1}, under Assumption \ref{ass0},  the matrix $H$ is symmetric and positive definite.
Differentiating (\ref{reeq3b}) along the trajectory of (\ref{reeq2}) gives
\begin{eqnarray}\label{reeq4b}
\dot{V}  &=&\tilde{\eta}^T\left(H\otimes I_m\right)\dot{\tilde{\eta}}+\mu_{2}^{-1}\tilde{\omega}^T\dot{\tilde{\omega}}\nonumber\\
&=&\tilde{\eta}^T\left(H\otimes S\left(\omega\right)\right)\tilde{\eta}-\mu_1\tilde{\eta}^T\left(H^2\otimes I_m\right)\tilde{\eta} \nonumber\\
&&+ \tilde{\eta}^T\left(H\otimes I_m\right)S_d\eta+\mu_{2}^{-1}\tilde{\omega}^T\dot{\tilde{\omega}}.
\end{eqnarray}
Since $S(\omega)$ is skew symmetric and $H$ is symmetric, $H\otimes S\left(\omega\right)$ is skew symmetric. Thus,
we have
\begin{eqnarray}\label{reeq4bb}
\dot{V} &=& -\mu_1\tilde{\eta}^T\left(H^2\otimes I_m\right)\tilde{\eta}- e_{v}^TS_d\eta+\mu_{2}^{-1}\tilde{\omega}^T\dot{\tilde{\omega}}\nonumber\\
&=&-\mu_1\tilde{\eta}^T\left(H^2\otimes I_m\right)\tilde{\eta}\nonumber\\ &&-\sum\limits_{i=1}^{N}e_{vi}^TS\left(\tilde{\omega}_i\right)\eta_i+\mu_{2}^{-1}\tilde{\omega}^T\dot{\tilde{\omega}}.
\end{eqnarray}
By Proposition \ref{Proposition2}, we  have, for $i=1,\cdots,N$,
$$e_{vi}^TS\left(\tilde{\omega}_i\right)\eta_i=\tilde{\omega}_i^T\phi(e_{vi})\eta_i,$$  which can be put in the following compact form:
\begin{eqnarray}
e_{v}^TS_d (\tilde{\omega}) \eta  = \tilde{\omega}^T\phi_d(e_{v})\eta.
\end{eqnarray}
Thus, from (\ref{reeq4bb}), we have
\begin{eqnarray}
\dot{V}&=&-\mu_1\tilde{\eta}^T\left(H^2\otimes I_m\right)\tilde{\eta}-\tilde{\omega}^T\phi_d(e_{v})\eta+\mu_{2}^{-1}\tilde{\omega}^T\dot{\tilde{\omega}}\nonumber\\
&=&-\mu_1\tilde{\eta}^T\left(H^2\otimes I_m\right)\tilde{\eta}-\tilde{\omega}^T\left(\phi_d(e_{v})\eta-\mu_{2}^{-1}\dot{\tilde{\omega}}\right).
\end{eqnarray}
Using  (\ref{reeq2abbbb}) gives
\begin{eqnarray}\label{reeq4bbb}
\dot{V} &=& -\mu_1\tilde{\eta}^T\left(H^2\otimes I_m\right)\tilde{\eta}\leq0.
\end{eqnarray}
Since $V$ is positive definite and $\dot{V}$ is negative semi-definite, $V$ is bounded, which means $\tilde{\eta}$ and $\tilde{\omega}$ are bounded. From  (\ref{reeq2b}), $\dot{\tilde{\eta}}$ is bounded, which implies $\ddot{V}$ is bounded. By Barbalat's Lemma, we have
 $$\lim\limits_{t\rightarrow\infty}\dot{V}(t)=0,$$
 which implies (\ref{eq101a}). Thus, by (\ref{eqev}), we have
 $$ \lim\limits_{t\rightarrow\infty}e_{v}(t)=0,$$
  which together with (\ref{reeq2abbbb}) yields (\ref{eq101b}).
  To show (\ref{eq101c}), differentiating  $\dot{\tilde{\eta}}$ gives,
\begin{eqnarray}\label{reeq6b}
 \ddot{\tilde{\eta}} & = &\left(I_N\otimes S (\dot{\omega}) \right)\tilde{\eta} + \left(I_N\otimes S\left(\omega\right)-\mu_1H\otimes I_m\right) \dot{\tilde{\eta}}   \nonumber \\
&&+S_d (\dot{\tilde{\omega}})\eta + {S}_d (\tilde{\omega}) \dot{\eta}.
 \end{eqnarray}
We have shown  ${\tilde{\eta}}$, $\tilde{\omega}$,  and $\dot{\eta}_i$ are all bounded.
 From (\ref{reeq2}), $\dot{{\tilde{\eta}}}$ and $\dot{\tilde{\omega}}$ are also bounded. Thus, $\ddot{\tilde{\eta}}$ is bounded. By Barbalat's Lemma again, we have $\lim\limits_{t\rightarrow\infty}\dot{\tilde{\eta}}=0$, which together with (\ref{eq101a})
 implies (\ref{eq101c}).
\end{proof}
 \begin{rem}\label{vsxi}  As a result of Lemma \ref{lemma4}, we have, for $i=1,\cdots,N$,
\begin{subequations}\begin{align}
   \lim\limits_{t\rightarrow\infty} \left(C\eta_i-q_0\right)=& \lim\limits_{t\rightarrow\infty} C\left(\eta_i-v\right)=0,  \label{conzero}\\
     \lim\limits_{t\rightarrow\infty}  \left(C\dot{\eta}_i-\dot{q}_0\right)=&\lim\limits_{t\rightarrow\infty} C\left(S\left(\omega_i\right)\eta_i+\mu_1e_{vi}-\dot{v}\right)\nonumber\\
=& \lim\limits_{t\rightarrow\infty} C\left(S\left(\tilde{\omega}_i\right)\eta_i+\mu_1e_{vi}+S\left(\omega\right)\tilde{\eta}_i\right) \nonumber\\
    =&0.\label{dotconzero}
 \end{align}
\end{subequations}
\end{rem}
Lemma \ref{lemma4} does not guarantee
$\lim\limits_{t\rightarrow\infty} \tilde{\omega} =0$. It is possible to make $\lim\limits_{t\rightarrow\infty} \tilde{\omega}=0$ if the signal
$v (t)$ is persistently exciting in the following sense.

\begin{defi}\cite{r18} A bounded piecewise continuous function $f:[0,+\infty)\mapsto \mathds{R}^n$ is said to be persistently exciting (PE) if there exist positive constants $\epsilon$, $t_0$, $T_0$ such that,
$$\frac{1}{T_0}\int^{t+T_0}_{t} f(s) f^T (s) ds\geq\epsilon I_n,~~~~\forall t\geq t_0$$
\end{defi}

\begin{rem}\label{lemmapeini}
A  piecewise continuous function $f:[0,\infty)\mapsto \mathds{R}^n$ is said to have spectral lines at frequencies $\omega_1,\cdots,\omega_n$, if,
for $k=1,\cdots,n$,
$$\lim\limits_{\delta\rightarrow\infty}\frac{1}{\delta}\int_{t}^{t+\delta}f(s)e^{-jw_ks}ds=\hat{f}(w_k)\neq0,$$
uniformly in $t$ \cite{boyd}. It was established in Lemma 3.4 of \cite{boyd}  that $f(t)$ is PE if $\hat{f}(w_k)$, $k=1,\cdots,n$, are linearly independent over $\mathds{C}^n$.
\end{rem}

We also need the following result which is taken from Lemma 2.4 of \cite{chen2015stabilization}.

\begin{lem}\label{lemmape} Consider a continuously differentiable function $g:[0,+\infty)\mapsto \mathds{R}^n$ and a bounded piecewise continuous function $f:[0,+\infty)\mapsto \mathds{R}^n$, which satisfy
$$\lim_{t\rightarrow \infty}g^T(t)f(t)=0.$$
Then, $\lim_{t\rightarrow \infty}g(t)=0$ holds under the following two conditions:
\begin{enumerate}
  \item[(i)]  $\lim_{t\rightarrow \infty}\dot{g}(t)=0$
  \item[(ii)]  $f(t)$ is persistently exciting.
\end{enumerate}
\end{lem}

To ensure that the leader's state is PE, we need to strengthen Assumption \ref{ass1} to the following:

\begin{ass}\label{ass2} For all $\omega$, all the eigenvalues of $S(\omega)$ are simple with zero real part.
\end{ass}

\begin{rem} Under Assumption \ref{ass2},  if $m$ is even, then
  we can  assume that $S (\omega)$ takes the following form:
\begin{eqnarray}\label{skewe}
  S(\omega)&=& \diag(\omega)\otimes a;
\end{eqnarray}
and, if $m$ is odd, then we can  assume that $S (\omega)$ takes the following form:
\begin{eqnarray}\label{skew}
  S(\omega)&=& \left[
              \begin{array}{cc}
               0 &  0 \\
                0  & \diag({\omega})\otimes a \\
              \end{array}
            \right],
\end{eqnarray}
For convenience,  we assume $m$ is even.  The case where $m$ is odd can be studied in a way similar to the case where $m$ is even.
\end{rem}

\begin{lem}\label{lemccci}  Under Assumptions  \ref{ass0} and \ref{ass2}, if, for all $i = 1, \cdots, l$ with $l = \frac{m}{2}$,   $\col (v_{2i-1} (0), v_{2i} (0)) \neq 0$,  then,\\
(i) $v (t)$ is PE; and \\
(ii)
\EQ \label{tomega}
\lim\limits_{t\rightarrow\infty} \tilde{\omega}=0.
\EN
\end{lem}
\begin{proof}
Part (i) \\
For each $ i = 1, \cdots, \frac{m}{2}$,
\EQQ
v (t) = \left [\begin{array}{c}
C_1 \sin (\omega_{01} t + \psi_1) \\
C_1 \cos (\omega_{01} t + \psi_1) \\
\vdots \\
C_{l} \sin (\omega_{0l} t + \psi_l) \\
C_{l} \cos (\omega_{0l} t + \psi_l)
\end{array} \right ]
\ENN
where,   for $i = 1, \cdots, l$, $C_i = \sqrt{v^2_{2i-1} (0) + v^2_{2i} (0)}$ and $ \tan \psi_i =  \frac{v_{2i-1} (0)}{v_{2i} (0)}$. Let $f_i(t)=\col (v_{2i-1} (t), v_{2i} (t))$. Then, we have,
for $i=1,\cdots,\frac{m}{2}$,
\EQQ f_i(t)&=&\left[
                                                                      \begin{array}{c}
                                                                        v_{2i-1} (t) \\
                                                                       v_{2i} (t) \\
                                                                      \end{array}
                                                                    \right]\\
&=&\left[
                                                                      \begin{array}{c}
                                                                        C_i \sin (\omega_{0i} t + \psi_i) \\
                                                                       C_i \cos (\omega_{0i} t + \psi_i) \\
                                                                      \end{array}
                                                                    \right]\\
&=&\left[
                                                                                   \begin{array}{c}
                                                                                      C_i\frac{e^{ \left(\omega_{0i} t+\psi_i\right)j}-e^{ -\left(\psi_i+\omega_{0i} t\right)j}}{2j} \\
                                                                                      C_i\frac{e^{ \left(\omega_{0i} t+\psi_i\right)j}+e^{ -\left(\omega_{0i} t+\psi_i\right)j}}{2} \\
                                                                                   \end{array}
                                                                                 \right]\\
&=&\frac{C_i}{2}\left[
                                                                                   \begin{array}{c}
                                                                                      e^{ -\psi_ij}e^{-\omega_{0i}j t }j-e^{ \psi_ij}e^{\omega_{0i}j t }j \\
                                                                                      e^{ \psi_ij}e^{\omega_{0i}j t }+e^{ -\psi_ij}e^{-\omega_{0i}j t } \\
                                                                                   \end{array}
                                                                                 \right],
                                                                                 \ENN

where $j=\sqrt{-1}$. It can be verified that, for $i=1,\cdots,\frac{m}{2}$,
\begin{subequations} \begin{align}
   \hat{f}_i\left(\omega_{0i}\right) &= \lim\limits_{\delta\rightarrow\infty}\frac{1}{\delta}\int_{\epsilon}^{\epsilon+\delta}{f_{i}(t)e^{-\omega_{0i}jt}}dt\nonumber \\
   &=\frac{C_i}{2}\col\left(-e^{ \psi_ij}j,e^{ \psi_ij}\right) =\frac{C_ie^{ \psi_ij}}{2}\col\left(-j ,1\right), \nonumber \\
    \hat{f}_i\left(-\omega_{0i}\right)& = \lim\limits_{\delta\rightarrow\infty}\frac{1}{\delta}\int_{\epsilon}^{\epsilon+\delta}{f_{i}(t)e^{\omega_{0i}jt}}dt\nonumber \\
    &=\frac{C_i}{2}\col\left(e^{ -\psi_ij}j,e^{ -\psi_ij}\right)=\frac{C_ie^{ -\psi_ij}}{2}\col\left(j,1\right).\nonumber
 \end{align}
\end{subequations}
Under Assumption \ref{ass2},  for all $i = 1, \cdots, l$, $\omega_i$ are distinct.
Thus, for $i=1,\cdots,\frac{m}{2}$, $\hat{f}_{i}\left(\omega_{0i}\right)$ and $\hat{f}_{i}\left(-\omega_{0i}\right)$  are linearly independent over $\mathds{C}^2$ if and only if $C_i\neq0$.
Since, none of $C_i$ is equal to zero due to our assumption on
the initial condition $v (0)$, by Remark \ref{lemmapeini}, for $i=1,\cdots,\frac{m}{2}$, $f_i(t)=\col (v_{2i-1} (t), v_{2i} (t))$ is PE.
Finally, noting that, for all $i, j=1,\cdots,\frac{m}{2}$ and $ i \neq j$,
$$\frac{1}{T_0}\int^{t+T_0}_{t} f_i(s) f_j^T (s) ds = 0_{2\times 2},~~~~\forall t\geq t_0$$ concludes that
$v (t) = \mbox{col } (f_1 (t), \cdots, f_l (t)$ is PE.

 Part (ii) \\ Note that the satisfaction of Assumption \ref{ass2} implies that of Assumption \ref{ass1}. Thus,
 (\ref{eq101a}) and (\ref{eq101c}) hold, which implies, for $i = 1, \cdots, N$,
\EQQ
&&\lim\limits_{t\rightarrow\infty} S\left(\tilde{\omega}_i\right) v (t) \\
&=& \lim\limits_{t\rightarrow\infty} S\left(\tilde{\omega}_i\right) ({\eta}_i - \tilde{\eta}_i) \\
&=& \lim\limits_{t\rightarrow\infty} S\left(\tilde{\omega}_i\right) {\eta}_i = 0.
\ENN
Since $v (t)$ is PE, by Lemma \ref{lemmape}, we have (\ref{tomega}).
\end{proof}
  \begin{rem}\label{remwu} Reference \cite{wu2017adaptive} introduced an adaptive reference generator for the uncertain exosystem  (\ref{leaderall})
with
\begin{eqnarray*}
 C &=& \left[
                   \begin{array}{cccc} 1 & 0 &\cdots& 0 \end{array}
                 \right],
\end{eqnarray*}
and
\begin{eqnarray*}
  S(\rho) &=& \left[
                   \begin{array}{c|c}
                     0 & I_{m-1} \\
                     \hline
                     -s_{1}  &-s_{2} ,\cdots,-s_{m}  \\
                   \end{array}
                 \right],
\end{eqnarray*}
where $\rho = \col ( s_{1}  s_{2}, \cdots, s_{m})$ is in a known compact set $\mathcal{Q} \subset \mathds{R}^{m}$,  $S(\rho)\in \mathds{R}^{m\times m}$ satisfies Assumption \ref{ass1}
for all $\rho \in \mathcal{Q}$.

Under the assumption (Assumption 9 of \cite{wu2017adaptive})
that there exists a positive definite matrix $P\in \mathds{R}^{m\times m}$, independent of $\rho $,  that satisfies, for all $\rho  \in\mathcal{Q}$,  the following equation
\begin{eqnarray}\label{areqa}
P^TS(\rho)+S(\rho)P-2xPC^TCP+2b I_m<0,
\end{eqnarray}
where $b,x$ are some positive numbers,  reference \cite{wu2017adaptive} proposed, in  terms of our notation, a dynamic compensator of the following form:
\begin{subequations} \label{eqaobserver} \begin{align}
   \dot{\eta}_i &= \left[
                   \begin{array}{c|c}
                     0 & I_{m-1} \\
                     \hline
                     -\hat{s}_{1,i} &-\hat{s}_{2,i},\cdots,-\hat{s}_{m,i}  \\
                   \end{array}
                 \right] \eta_i+PC^TCe_{vi},\label{eqaobservera}\\
    \dot{\hat{s}}_{i,k} &=  -\left(e_{vi}^TP^{-1}\right)^{(m)}\eta_{i}^{(k)},\label{eqaobserverb}
 \end{align}
\end{subequations}
 where, for $k=1,\cdots,m$ and $i=1,\cdots,N$,  $\left(e_{vi}^TP^{-1}\right)^{(m)}$ denotes the last element of the row vector $e_{vi}^TP^{-1}$, the $\eta_{i}^{(k)}$ denotes the $kth$ element in the column vector $\eta_i$. It was shown that, under some additional assumptions,  the dynamic compensator (\ref{eqaobserver}) can also estimate the state of (\ref{leaderall}) and the unknown parameters
 $s_{i} $, $ i = 1, \cdots, m$.  Nevertheless, since such matrix $P$ is generally a function of $\rho$, it is unlikely that
  Assumption 9 of \cite{wu2017adaptive} can be satisfied for all $\rho$ in a given compact set $\mathcal{Q}$.

 Other advantages of our adaptive distributed observer (\ref{compensator2})  include  that we only need to estimate $\frac{m}{2}$ unknown parameters by $\frac{mN}{2}$ equations instead of $m$ parameters by $m \times N$ equations as in (\ref{eqaobserver})),
 we do not limit our unknown parameters to lie in some known compact set $\mathcal{Q}$, and our $C$ matrix can have multiple rows.
 \end{rem}

\section{Leader-following Consensus of Multiple EL Systems} \label{section4}
In this section, we will apply our adaptive distributed observer (\ref{compensator2})  to synthesize a distributed control law to solve the leader-following consensus problem of multiple Euler-Lagrange systems
with an uncertain leader system. As in \cite{cai2016leader}, for $i=1,\cdots,N$, let
\begin{subequations}\label{claw}\begin{align}
\dot{q}_{ri}&=CS\left(\omega_i\right)\eta_i-\alpha\left(q_i-C\eta_i\right), \label{claw1}\\
 s_i&=\dot{q}_i-\dot{q}_{ri},\label{claw2}
 \end{align}
\end{subequations}
where $\alpha$ is positive number.

We now propose our control law as follows:
\begin{subequations}\label{dstureq11}
 \begin{align}
  \tau_i&=-K_{i}s_i+Y_i\hat{\Theta}_i,\label{dstureq11i}\\
   \dot{\hat{\Theta}}_i&=-\Lambda^{-1}_i Y_{i}^Ts_i,\label{dstureq11ii}\\
  \dot{\eta}_i &=S\left(\omega_i\right)\eta_i + \mu_1\sum\limits_{j \in \mathcal{\bar{N}}_i}(\eta_j-\eta_i),\label{compensator2ccb}\\
\dot{\omega}_{i} &=\mu_2\phi(e_{vi})\eta_i,
 \end{align}
\end{subequations}
where, for $i=1,\cdots,N$,  $K_i$ and $\Lambda_i$ are positive definite matrices, $\hat{\Theta}_i$ are vectors for  estimating $\Theta_i$,
 $\Lambda_i$ are  diagonal matrices with positive diagonal entries which determine the update rates of $\hat{\Theta}_i$, and
the matrix function $\phi(\cdot)$ is as defined in (\ref{skewopen}).

\begin{rem} Since the system matrix $S$ of our leader system is not known exactly, the third equations of our control law contain $N$  unknown vectors $\omega_i$, which are updated
by the fourth equations of (\ref{dstureq11}). In contrast, in \cite{cai2016leader}, the system matrix $S (\omega)$ of the leader system must be known to all the children of the leader.
Therefore, the control law in \cite{cai2016leader} relies on the exact knowledge of the matrix $S (\omega)$.
\end{rem}
\begin{thm}Consider systems (\ref{MARINEVESSEL1}), (\ref{leaderall}) and the graph $\bar{\mathcal{G}}$. Under Assumptions \ref{ass0} and \ref{ass1}, Problem \ref{ldlesp} is solvable by the control law (\ref{dstureq11}).
\end{thm}
\begin{proof}
Differentiating (\ref{claw}) gives, for $i=1,\cdots,N$,
\begin{subequations}\label{dclaw}\begin{align}
\ddot{q}_{ri}&=CS\left(\omega_i\right)\dot{\eta}_i+CS\left(\dot{\omega}_i\right)\eta_i-\alpha\left(\dot{q}_i-C\dot{\eta}_i\right),\label{dclaw1}\\
 \dot{s}_i&=\ddot{q}_i-\ddot{q}_{ri}. \label{dclaw2}
 \end{align}
\end{subequations}
By Property \ref{prop1}, there exists a known matrix $Y_i=Y_i\left(q_i,\dot{q}_i, \ddot{q}_{ri}, \dot{q}_{ri} \right)$ and an unknown constant vector $\Theta_i$ such that,
\begin{equation}\label{claw0}
 Y_i\Theta_i=M_i\left(q_i\right)\ddot{q}_{ri}+C_i\left(q_i,\dot{q}_i\right)\dot{q}_{ri}+G_i\left(q_i\right).
\end{equation}
Substituting $Y_i\Theta_i$ into (\ref{MARINEVESSEL1}) gives
\begin{eqnarray}\label{undistur2}
  M_i\left(q_i\right)\ddot{q}_i+C_i\left(q_i,\dot{q}_i\right)\dot{q}_i+G_i\left(q_i\right) -M_i\left(q_i\right)\ddot{q}_{ri}&&\nonumber\\
 -C_i\left(q_i,\dot{q}_i\right)\dot{q}_{ri}-G_i\left(q_i\right)+Y_i\Theta_i=&\tau_i.&
\end{eqnarray}
Using  (\ref{claw2}) and (\ref{dclaw2})  in (\ref{undistur2}) gives
\begin{eqnarray}\label{undistur3}
M_i\left(q_i\right)\dot{s}_i=-C_i\left(q_i,\dot{q}_i\right)s_i+\tau_i-Y_i\Theta_i.\label{undistur1i}
\end{eqnarray}

Substituting (\ref{dstureq11i}) and (\ref{dstureq11ii}) into (\ref{undistur3}) gives, for $i=1,\cdots,N$,
 \begin{subequations}\label{MARINEVESSEL14}
\begin{align}
 M_i\left(q_i\right)\dot{s}_i&=-C_i\left(q_i,\dot{q}_i\right)s_i-K_is_i+Y_i\tilde{\Theta}_i,\\
  \dot{\hat{\Theta}}_i&=\Lambda_{i}^{-1}Y_{i}^Ts_i,\label{MARINEVESSEL14ii}
\end{align}
\end{subequations}
where $\tilde{\Theta}_i=\left(\hat{\Theta}_i-\Theta_i\right)$.

For $i=1,\cdots,N$,  let \begin{equation}\label{MARINEVESSEL15}
V_i=\frac{1}{2}\left(s_{i}^TM_i\left(q_i\right)s_i+\tilde{\Theta}_{i}^T\Lambda_{i}\tilde{\Theta}_{i}\right),
\end{equation}
Then, along the trajectory of (\ref{MARINEVESSEL14}),
\begin{eqnarray*} 
  \dot{V}_i&=&s_{i}^TM_i\left(q_i\right)\dot{s}_i+\frac{1}{2}s_{i}^T\dot{M}_i\left(q_i\right)s_i+\tilde{\Theta}_{i}^T\Lambda_{i}\dot{\tilde{\Theta}}_{i}\nonumber\\
  &=&s_{i}^T\left(-C_i\left(q_i,\dot{q}_i\right)s_i-K_is_i+Y_{i}\tilde{\Theta}_i\right)\nonumber\\
  &&+\frac{1}{2}s_{i}^T\dot{M}_i\left(q_i\right)s_i+\tilde{\Theta}_{i}^T\Lambda_{i}\dot{\tilde{\Theta}}_{i}.\nonumber \\
 \end{eqnarray*}
Noting $\left(\dot{M}_i\left(q_i\right)-2C_i\left(q_i,\dot{q}_i\right)\right)$ is skew symmetric gives
\begin{eqnarray}\label{MARINEVESSEL16}
  \dot{V}_i&=&-s_{i}^TK_is_i, ~i=1,\cdots,N.
  \end{eqnarray}
 Since $K_i$ are positive definite,
 the vectors $s_{i}$ and $\tilde{\Theta}_i$ are bounded. We now further show $s_i\rightarrow0$ as $t\rightarrow\infty$ using Barbalat's lemma.
For this purpose, substituting (\ref{claw1}) into (\ref{claw2}) gives
\begin{equation}\label{MARINEVESSEL17}
 \dot{q}_i+\alpha q_i=s_i+CS\left(\omega_i\right)\eta_i+\alpha C\eta_i.
 \end{equation}
Since $s_i$ is already shown to be bounded, $S\left(\omega_i\right)$ and $\eta_i$ are also bounded by Lemma \ref{lemma4}, and $ \alpha  $ is positive, equation (\ref{MARINEVESSEL17}) can be viewed as a stable first order linear system in $q_i$ with a bounded input. Thus, both $q_i$ and $\dot{q}_i$ are bounded. Thus, from (\ref{claw1})  and
(\ref{dclaw1}), $\dot{q}_{ri}$ and $\ddot{q}_{ri}$ are bounded.
By equation (\ref{dclaw2}), $\dot{s}_i$ are also bounded. Thus $\ddot{V}_i$ are bounded, which implies $\dot{V}_i$ is uniformly continuous. By Barbalat's lemma, we have, for $i=1,\cdots,N$, $\dot{V}_i\rightarrow 0$ as $t\rightarrow \infty$, which implies  $s_i\rightarrow 0$ as $t\rightarrow \infty$.

Using (\ref{compensator2ccb}) in  (\ref{claw}) gives
   \begin{eqnarray}\label{MARINEVESSEL19}
   \dot{q}_{i}-C\dot{\eta}_i+\alpha\left(q_{i}-C\eta_i\right) =s_i-\mu_1C\sum\limits_{j \in \mathcal{\bar{N}}_i}(\eta_j-\eta_i).
  \end{eqnarray}
With $e_i=\left(q_i-C\eta_i\right)$, equation (\ref{MARINEVESSEL19}) can be rewritten as
   \begin{eqnarray}\label{reMARINEVESSEL19}
   &&\dot{e}_{i}+\alpha e_i  = s_i-\mu_1Ce_{vi}.
  \end{eqnarray}
Since,  by Lemma \ref{lemma4}, $\lim_{t\rightarrow\infty}e_{vi}= 0$, for $i=1,\cdots,N$, equation (\ref{reMARINEVESSEL19}) can be viewed as a stable first order differential equation in $e_i$ with
$\left(s_i-\mu_1Ce_{vi}\right)$ as the input, which are bounded over $t\geq 0$ and tend to zero as $t\rightarrow \infty$. Thus,  we conclude that both $e_i=\left(q_i-C\eta_i\right)$ and $\dot{e}_i=\left(\dot{q}_i-C\dot{\eta}_i\right)$ are bounded over $t\geq 0$ and will tend to zero as $t\rightarrow \infty$. These facts together with (\ref{conzero}) and (\ref{dotconzero}) complete the proof.
\end{proof}

\section{An Example}\label{section5}

Consider a group of six Euler-Lagrange systems, each of which describes a two-link robot arm whose motion equations taken from \cite{lewis1993control} are as follows:
\begin{equation}\label{numerica1}
  M_i\left(q_i\right)\ddot{q}_i+C_i\left(q_i,\dot{q}_i\right)\dot{q}_i+G_i\left(q_i\right)=\tau_i, i=1,\cdots,6,
\end{equation}
where $q_i=\col\left(q_{i1}, q_{i2}\right)$,
\begin{align}
  M_i\left(q_i\right) &= \left(
                            \begin{array}{cc}
                              a_{i1} +a_{i2}+2a_{i3}\cos q_{i2} &  a_{i2}+a_{i3}\cos q_{i2} \\
                               a_{i2}+ a_{i3}\cos q_{i2}&  a_{i2} \\
                            \end{array}
                          \right),\nonumber \\
  C_i\left(q_i,\dot{q}_i\right) &= \left(
                                      \begin{array}{cc}
                                        -a_{i3}\left(\sin q_{i2}\right)\dot{q}_{i2}& -a_{i3}\left(\sin q_{i2}\right)\left(\dot{q}_{i1}+\dot{q}_{i2}\right) \\
                                         a_{i3}\left(\sin q_{i2}\right)\dot{q}_{i1} & 0  \\
                                      \end{array}
                                    \right),
  \nonumber\\
 G_i\left(q_i\right) &=\left(
                          \begin{array}{c}
                            a_{i4}g\cos q_{i1}+a_{i5}g\cos\left( q_{i1}+ q_{i2}\right) \\
                            a_{i5}g\cos\left( q_{i1}+ q_{i2}\right) \\
                          \end{array}
                        \right),
 \nonumber
\end{align}
and the unknown vector $\Theta_i=\col \left(a_{i1},a_{i2},a_{i3},a_{i4},a_{i5}\right)$. The communication topology is shown in Figure.\ref{fig1}, which satisfies Assumption \ref{ass0}. The leader's signal is
$$q_0=\col\left(A_1 \sin(\omega_{01} t+ \phi_1), A_2 \sin(\omega_{02} t + \phi_2) \right),$$ where $A_1$, $A_2$, $\omega_{01}$ and $\omega_{02}$ are  arbitrary unknown positive real numbers;  and $\phi_1$ and $\phi_2$ are
arbitrary unknown real numbers.
This leader's signal can be
generated by (\ref{leaderall}) with
\begin{eqnarray*}
S=\left[
        \begin{array}{cccc}
          0 & \omega_{01} &0 &0\\
          -\omega_{01} &0 &0 & 0\\
            0 &0 & 0 & \omega_{02}\\
          0 & 0 & -\omega_{02} &0
        \end{array}
      \right],~C=\left[
        \begin{array}{cccc}
          1 & 0 & 0 & 0\\
          0 & 0 & 1  & 0\\
        \end{array}
      \right].
\end{eqnarray*}
Thus, Assumption \ref{ass2} is also satisfied.

By Lemma \ref{lemma4}, for $i=,1,\cdots,6$, we can design an adaptive distributed observer for (\ref{leaderall}) as follows:
\begin{subequations}\label{compensatorex}
\begin{align}
\dot{\eta}_i &=S\left(\omega_i\right)\eta_i + \mu_1\sum_{j \in \mathcal{\bar{N}}_i}(\eta_j-\eta_i),\label{compensator2bex}\\
\dot{\omega}_{i} &=\mu_2 \phi(e_{vi})\eta_i,~ i = 1, \cdots, 6, \label{compensator2cex}
\end{align}
\end{subequations}

where $\mu_1=10$, $\mu_2=20$ and $\omega_{i} \in \mathds{R}^2$.

Based on this observer, we can further synthesize a control law of the form (\ref{dstureq11}) with the following parameters: $K_i=20I_2$, $\alpha=10$, and $\Lambda_i=10$. For $i=1,\cdots,N$, the actual value of $\Theta_i$ are given as follows:
\begin{eqnarray*}
 && \Theta_1 = \col(0.64,1.10, 0.08,0.64,0.32), \\
  && \Theta_2 =\col(0.76,1.17,0.14,0.93,0.44),\\
  &&\Theta_3 = \col(0.91,1.26,0.22, 1.27,0.58),\\
  &&\Theta_4 = \col(1.10,1.36,0.32,1.67,0.73),\\
  &&\Theta_5 = \col(1.21,1.16,0.12,1.45, 1.03),\\
   &&\Theta_6 = \col(1.31,1.56,0.22,1.65,1.33).
\end{eqnarray*}
Simulation is conducted with the following initial condition: $q_{i}(0)=0$, $\hat{\Theta}_i(0)=0$, randomly chosen $\omega_i(0)$ in the range $[0,1]$, $i=1,\cdots,6$. The leader's initial condition is $v_0(0)=\col(1,0, 1,0)$,  and the actual unknown frequencies are $\omega_{01} = 4$ and $\omega_{02} = 2$.
 Figure \ref{fig2} shows the trajectories of $q_i$ and $\dot{q}_i$, respectively, for $i=1,\cdots,6$. Figure \ref{fig4} shows the tracking performance of $q_i$ and $\dot{q}_i$, for $i=1,\cdots,6$.

Since $v_0(0)=\col(1,0, 1,0)$, by Lemma \ref{lemccci}, $v (t)$ is PE. Thus,
  the distributed observer (\ref{compensator2}) will guarantee $\tilde{\omega}_i\rightarrow0$ as $t\rightarrow\infty$, for $i=1,\cdots,6$. Figures \ref{fig6} and \ref{fig5}
   show the trajectories of the two components of $\omega_i$ and $\tilde{\omega}_i$, respectively, for $i=1,\cdots,6$.

\begin{figure}
\centering
\setlength{\unitlength}{0.75mm}
\includegraphics[clip,scale=0.6]{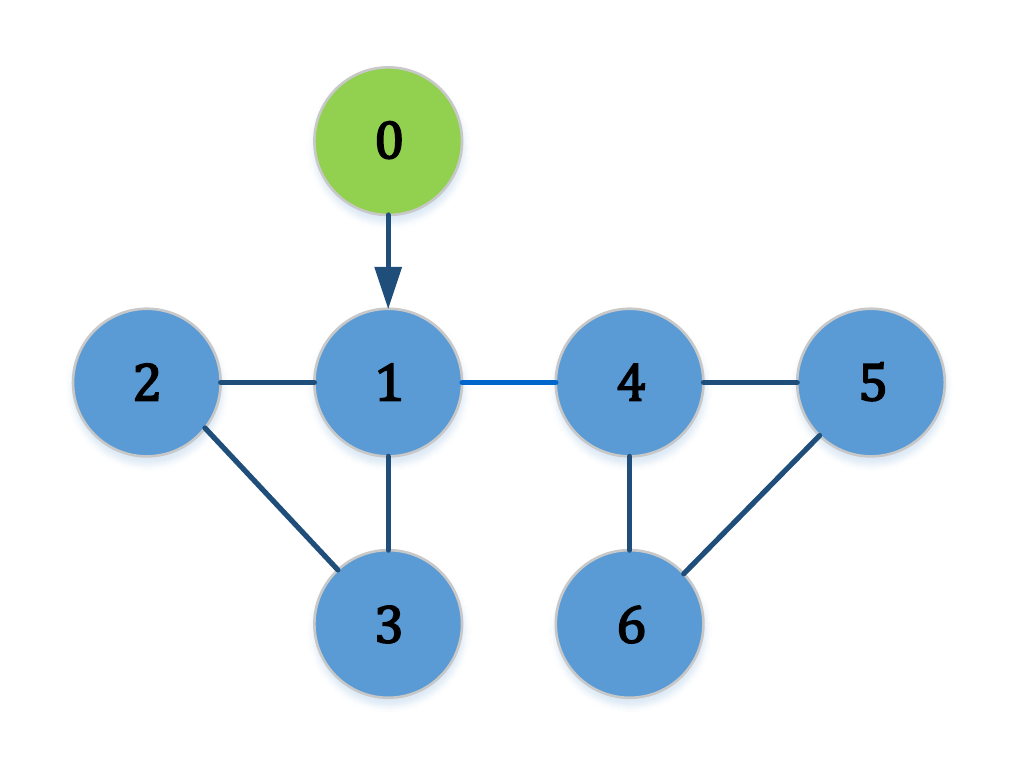}
\caption{ Communication topology $\bar{G}$}\label{fig1}
\end{figure}
\begin{figure}[ht]

  \centering
\epsfig{figure=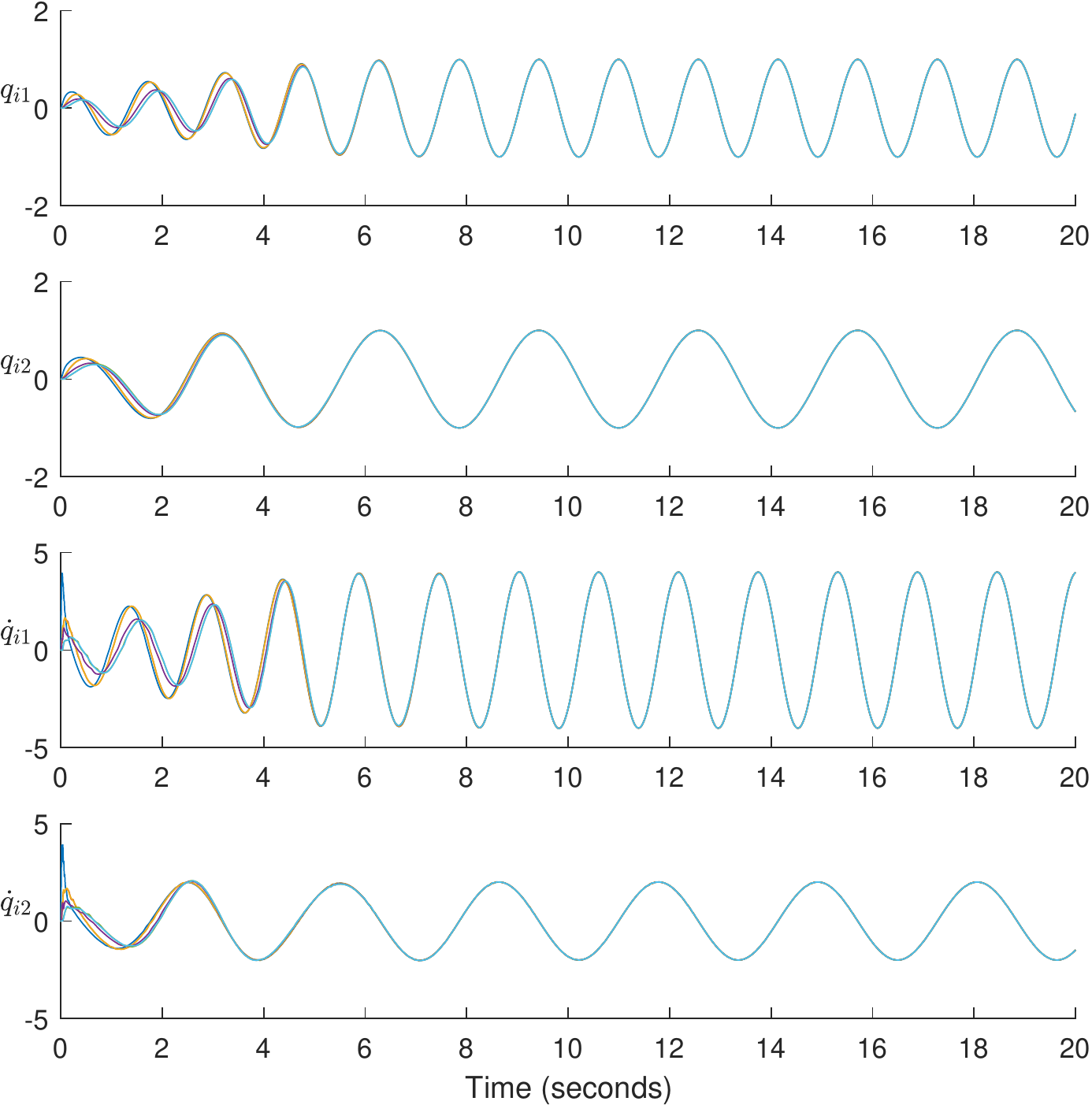,height=3.6in}

  \caption{Trajectory of $q_i$ and $\dot{q}_i$,~ $i = 1, \cdots, 6.$}\label{fig2}
\end{figure}

\begin{figure}[ht]

  \centering
  \epsfig{figure=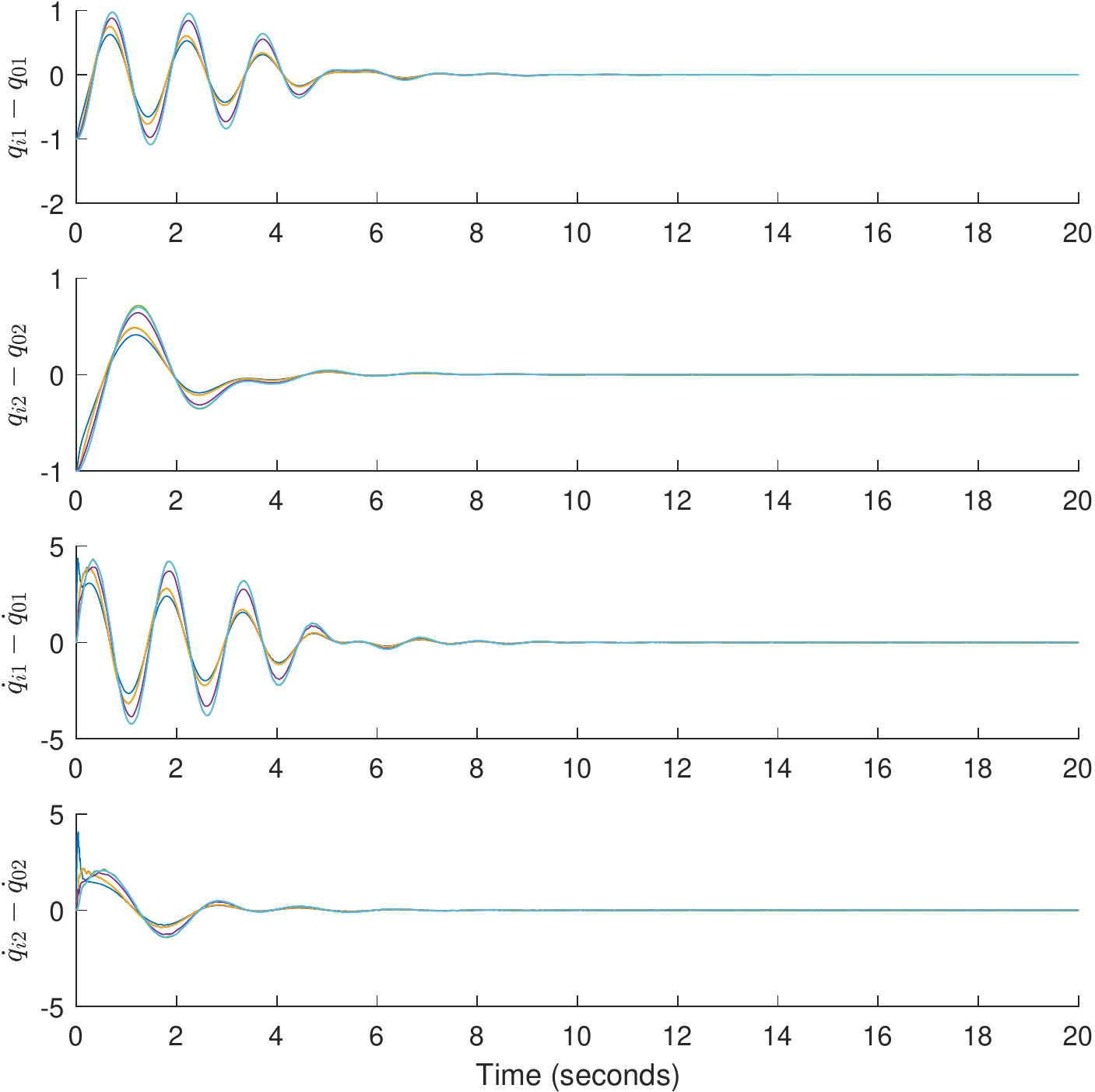,height=3.6in}

  \caption{Tracking Performance of $q_{i}$ and $\dot{q}_i$,~ $i = 1, \cdots, 6.$}\label{fig4}
\end{figure}
\begin{figure}[ht]
  \centering
  \epsfig{figure=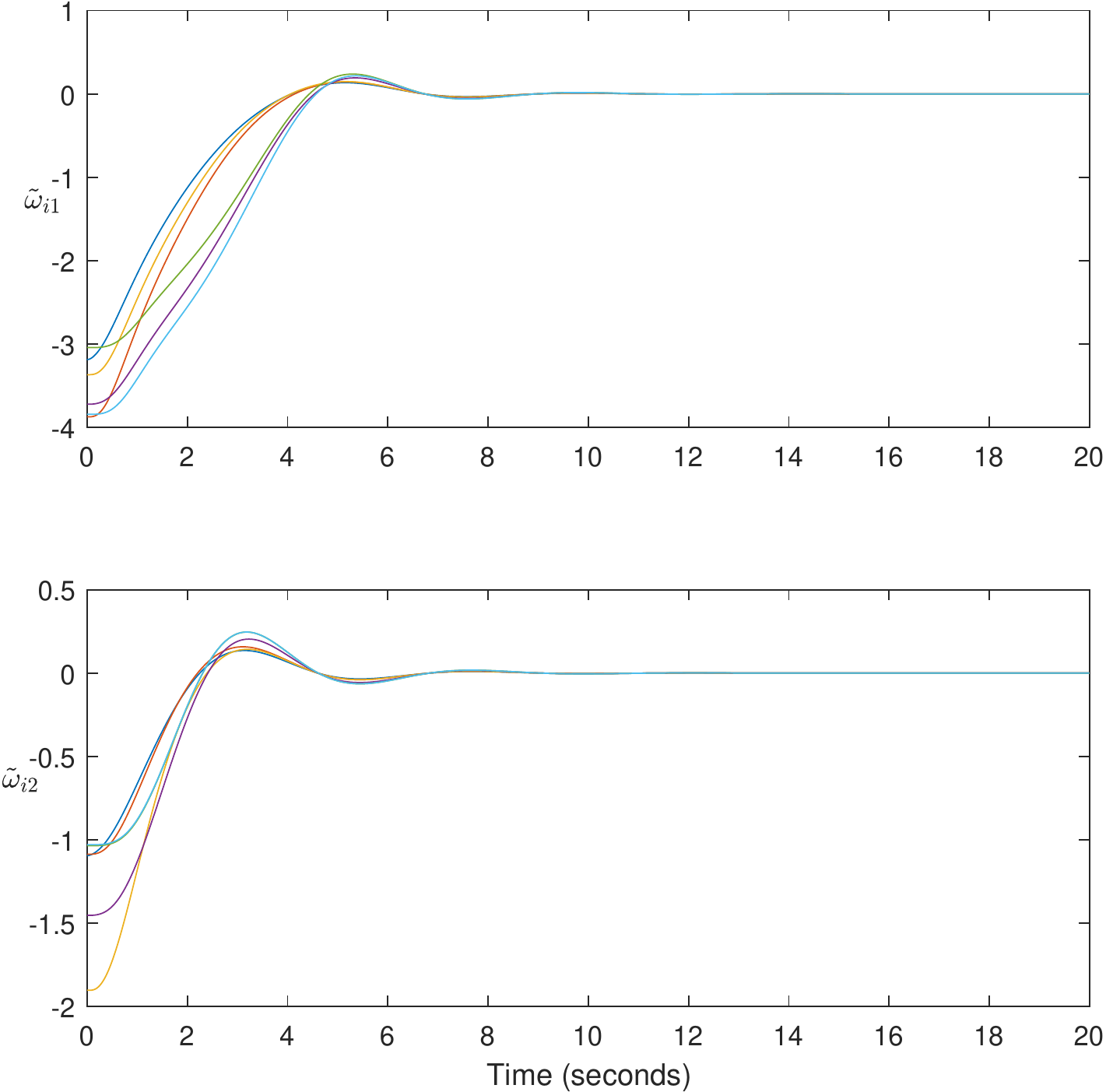,height=3.6in}

  \caption{Trajectory of $\tilde{\omega}_i$,~ $i = 1, \cdots, 6.$}\label{fig6}
\end{figure}
\begin{figure}[ht]
  \centering
  \epsfig{figure=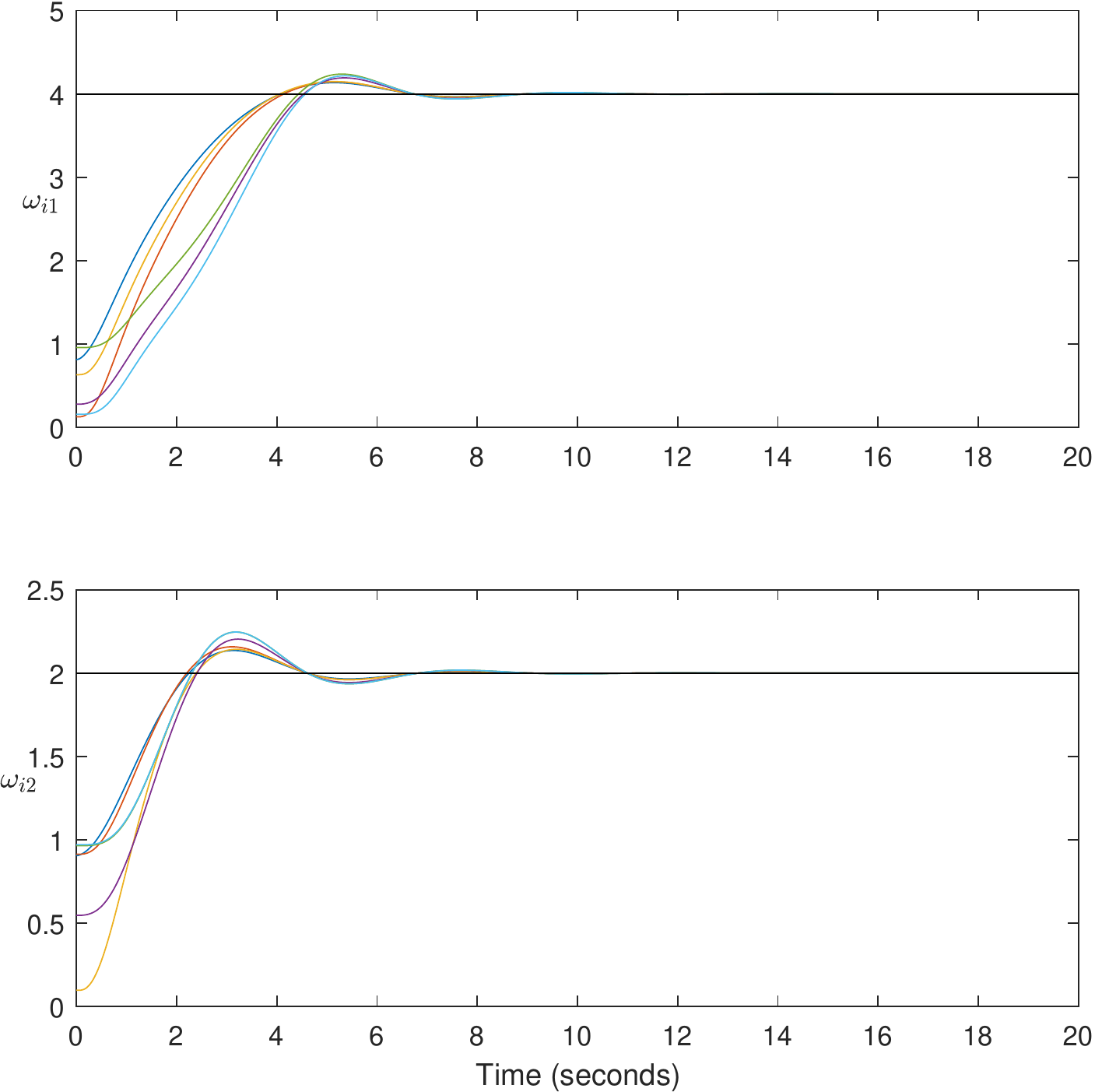,height=3.6in}

  \caption{Trajectory of $\omega_i$,~ $i = 1, \cdots, 6.$}\label{fig5}
\end{figure}

\section{Conclusion}\label{section6}

In this paper, we have studied the leader-following consensus problem of multiple Euler-Lagrange systems subject an uncertain leader system.
For this purpose,  we have first established an adaptive distributed observer for a neutrally  stable linear leader system whose system matrix is not known exactly. Under standard assumptions, this adaptive distributed observer can estimate and pass the leader's state to each follower through the communication network of the system without knowing the leader's system matrix. Further, if  the leader's state is persistently exciting,
this adaptive distributed observer can also asymptotically learn the unknown parameters of the leader's system matrix. On the basis of this adaptive distributed observer, we can synthesize an adaptive distributed control law to solve our problem via the certainty equivalence principle. Our result allows
the leader-following consensus problem of multiple  Euler-Lagrange systems to be solved even if none of our followers knows the exact information of the system matrix of the leader system.

The adaptive distributed observer for the uncertain leader system can also be applied to other problems such as the cooperative output regulation problem as studied in \cite{cai2017adaptive}.
In the future, we will further consider  establishing an adaptive distributed observer for an uncertain leader system subject to a switching topology.

\section{Appendix}
A brief introduction of graph theory is shown in the following which can be found in \cite{godsil2013algebraic}:

A digraph $\mathcal{G}=\left(\mathcal{V},\mathcal{E}\right)$ consists of a finite set of nodes $\mathcal{V}=\{1,2,\cdots, N\}$ and an edge set $\mathcal{E}=\{(i,j), i,j\in \mathcal{V}, i\neq j\}$. Weighted directed graph $\mathcal{G}$ is used to model communication among the $N$ systems. Graph $\mathcal{G}$ consists of a node set $\mathcal{V}=\{1,\cdots,N\}$, and edge set $\mathcal{E}\subseteq \left(\mathcal{V}\times \mathcal{V}\right)$ and a weighted adjacency matrix $\mathcal{C}=\left[c_{ij}\right]\in \mathds{R}^{N\times N}$ with $c_{ij}\geq 0$. If $c_{ij}>0$, then $\left(j,i\right)\in \mathcal{E}$. The in degree of node $i$ is defined as $d_i=\sum_{j\in\mathcal{ N}_i}c_{ij}$. Let $D=diag\{d_1,\cdots,d_N\}$ be the degree matrix of $\mathcal{G}$. The lapalcian matrix of $\mathcal{G}$ is defined as $H=D-\mathcal{C}$. A node $i$ is called a neighbor of a node $j$ if the edge $(i,j)\in \mathcal{E}$. $\mathcal{N}_i$ denotes the subset of $\mathcal{V}$ that consists of all the neighbors of the node $i$. If the graph $\mathcal{G}$ contains a sequence of edges of the form $\left(i_1, i_2\right)$, $\left(i_2, i_3\right)$, $\cdots$, $\left(i_k, i_{k+1}\right)$, then the set $\left\{\left(i_1, i_2\right),\left(i_2, i_3\right), \cdots, \left(i_k, i_{k+1}\right)\right\}$ is called a path of $\mathcal{G}$ from $i_1$ to $i_{k+1}$, and node $i_{k+1}$ is said to be reachable from node $i_1$.
%

\end{document}